\documentclass[12pt,a4paper]{amsart}
\usepackage[colorlinks=true,linkcolor=blue, citecolor=cyan, urlcolor=red]{hyperref}

\usepackage[english]{babel}
\usepackage{amssymb}
\usepackage{amsfonts}
\usepackage{amsthm}
\usepackage{amsmath}
\usepackage{latexsym}
\usepackage{color}

\newcommand{\R}{{\mathbb R}}

\newcommand{\Z}{{\mathbb Z}}
\newcommand{\N}{{\mathbb N}}

\theoremstyle{plain}
\newtheorem{theorem}{Theorem}[section]
\newtheorem{lemma}{Lemma}[section]
\newtheorem{proposition}{Proposition}[section]
\newtheorem*{proposition*}{Proposition}

\theoremstyle{definition}
\newtheorem{corollary}{Corollary}[section]
\newtheorem{remark}{Remark}[section]
\newtheorem*{remark*}{Remark}
\newtheorem{definition}{Definition}[section]

\theoremstyle{remark}
\newtheorem{example}{Example}[section]

\newcommand{\expectation}{\mathbb{E}}

\renewcommand{\P}{\mathbb{P}}
\newcommand{\G}{\mathcal{G}}

\newcommand{\K}{K}

\newcommand{\cav}{\mathrm{Cav\,}}

\newcommand{\val}{\mathrm{val}}

\newcommand{\M}{{M}}

\newcommand{\lip}{\mathrm{lip}}

\newcommand{\I}{I}
\newcommand{\J}{J}

\newcommand{\nr}{\mathrm{NR}}

\newcommand{\sign}{{\mathrm{sgn}}\,}
\newcommand{\dist}{{{d}}}
\newcommand{\argmin}{{\mathrm{argmin}}\,}

\numberwithin{equation}{section}

\begin{document}

\title[Repeated games with incomplete information and bounded values]{On repeated zero-sum games with incomplete information and asymptotically bounded values}

\thanks{
\indent \ $^1$ National Research University Higher School of Economics, Saint-Petersburg, Russia\\
 E-mail: \href{mailto:fsandomirskiy@hse.ru}{fsandomirskiy@hse.ru};  \href{mailto:sandomirski@yandex.ru}{sandomirski@yandex.ru} \\
\indent \ $^2$ St.~Petersburg Institute for Economics and Mathematics of Russian Academy of Sciences\\
\indent \  I am thankful to Vita~Kreps for many inspiring discussions and her care and to Misha~Gavrilovich for suggestions that significantly improved presentation of the results.\\
\indent \ Support from the Basic Research Program of the National Research University Higher School of Economics and from grants 13-01-00462,  13-01-00784, and 16-01-00269 of the Russian Foundation for Basic Research
is gratefully acknowledged.
}
%
%
%
%
%\vskip -0.5cm

\maketitle

\begin{center}
\vskip -0.7cm
{{Fedor Sandomirskiy}}$^{1,2}$
\end{center}

\vskip 0.7cm

\begin{flushright}
\emph{To the memory of Victor Domansky}
\end{flushright}

\vskip -0.0cm

%\author{Fedor Sandomirskiy}

%\institute{Department of Mathematical Physics, St.~Petersburg State University,  Ulianovskaja, 1,  St.~Petersburg-Petrodvoretz, 198904 Russia, \email{fedotov.s@mail.ru} \and 
%Chebyshev Laboratory, St.~Petersburg State University,
%              14th Line, 29b, Vasilyevsky Island, St.~Petersburg, 199178 Russia, \email{sandomirski@yandex.ru}
%}

%\institute{Fedor Sandomirskiy \at
%             1) Chebyshev Laboratory, Saint Petersburg State University,
%              14th Line, 29b, Vasilyevsky Island, Saint Petersburg, 199178 Russia\\
%2) Department of Mathematical Physics, Saint Petersburg State University, Ulianovskaja, 1, Petrodvoretz, Saint Petersburg, 198904 Russia
%      \\              Tel.: +79216332353\\
%              Fax: +78123636871\\
%              \email{sandomirski@yandex.ru}          

\begin{abstract}

We consider repeated zero-sum games with incomplete information on the side of Player~2 with the total payoff given by the non-normalized sum of stage gains. In the classical examples the value $V_N$ of such an $N$-stage game is of the order of $N$ or $\sqrt{N}$ as $N\to \infty$.

Our aim is to find what is causing another type of asymptotic behavior of the value $V_N$
observed for the discrete version of the financial market model introduced by De~Meyer and Saley. For this game Domansky and independently De~Meyer with Marino found that $V_N$ remains bounded as $N\to\infty$ and converges to the limit value. This game is almost-fair, i.e.,
if Player~1 forgets his private information the value becomes zero. 

We describe a class 
of almost-fair games having bounded values in terms of an 
 easy-checkable property of the auxiliary non-revealing game. 
We call this property the piecewise property,
and it says that
 there exists an optimal strategy of Player~2 that is piecewise-constant  as a function of a prior distribution $p$. 
Discrete market models have the piecewise property.
 We show that for non-piecewise almost-fair games with an additional non-degeneracy condition $V_N$ is of the order of~$\sqrt{N}$.
\end{abstract}

\section{Introduction}
The theory of repeated games with incomplete information originated in the early sixties from 
 reports of Aumann and Maschler to the United States Arms Control and Disarmament Agency (the reports 
were published as a book:  Aumann and Maschler~(1995)~\cite{A-M}).
Their goal was to develop a game-theoretical framework for
repeated disarmament negotiations between the USSR and USA. The main feature of that interaction
 was a strategic usage of information in a dynamic framework,
that is, when selecting an action now 
one has to care about information revealed by one's action and 
its affect on future behavior of the opponent. 

It turned out that the simplest zero-sum case with
only one side having private information is already very nontrivial. 
For the introduction to the theory of such games we refer to Aumann and Maschler~(1995)~\cite{A-M}, %~\cite{A-M}.
 Zamir~(1992)~\cite{Zamir_handbook}, 
Sorin~(2002)~\cite{Sorin}, %
and to Mertens, Sorin and Zamir~(2015)~\cite{Bigbook} %
for further reading.

In this paper we consider the classical setting mentioned above:
zero-sum repeated games with incomplete information 
on one side. In these games two players repeatedly play the same zero-sum game which is 
selected by chance before the first stage 
of a multistage interaction according to a prior distribution 
known to both players. 
The game selected is told to Player~1 only, so he knows the game he is playing while his uninformed opponent does not know. 
Multistage setting gives to the uninformed player an opportunity to guess what is the actual stage game 
by observing previous actions of Player~1. 
In turn, informed player has to balance between two sometimes opposite goals: 
to benefit from his private information at a current stage and 
to avoid  fast revelation so as to be able to benefit in the future.  
Complexity of optimal strategies usually prevents their explicit description, 
and forces qualitative and asymptotic analyses to be the main tools to study such games.

One of the main asymptotic problems is to describe the behavior of the value $V_N$ of an $N$-stage game for large $N$. 
This problem has attracted much attention because the asymptotic behavior of the value is related to the benefit that 
the informed player can get from his information.

Throughout the paper we assume that the total payoff equals
to the sum of stage payoffs (i.e., there is neither discounting nor normalization). 
Then, as it was shown by Aumann and Maschler, under some technical assumptions $V_N$ 
grows linearly up to an error term bounded by a constant times $\sqrt{N}$
(Aumann and Maschler~(1995)~\cite{A-M}; 
%. 
Gensbittel~(2015)~\cite{CavU}
and Neyman~(2013)~\cite{Neyman} 
for recent extensions).

The linear component disappears if we assume that all the strategic advantages or disadvantages 
in a repeated game arise from  information asymmetry. Formally this means that 
the auxiliary one-stage game where nobody is informed (the so-called \emph{non-revealing game}) is fair for any prior distribution, 
i.e., has zero value. We call repeated games with this property \emph{almost-fair}. For almost-fair games
the value $V_N$ can be interpreted as the value of information possessed by Player~1,
and  
it coincides with the error term and thus does not exceed $\sqrt{N}$ by the order of magnitude. 
The growth is sublinear because Player~1 loses
the information advantage from stage to stage 
revealing his private information by past actions.

The paper deals with the problem of finding a relation between the asymptotic behavior of the value
$V_N$ of an almost-fair repeated game and the strategic properties of a non-revealing game.
 Importance of such a relation comes from the fact that the non-reveling game 
is usually easy to solve. 
Our main results are easily checkable sufficient conditions for the two kinds of asymptotic behavior of the value: 
being asymptotically bounded and $\sqrt{N}$-growth. 
The conditions are nearly opposite of  each other,
and hence our results provide an almost complete characterization of the possible asymptotic behaviors. 
This gives a general explanation of many results from the literature.    

For a long time all known almost-fair games had $\sqrt{N}$-behavior of the value. 
Almost all methods in the theory of repeated games with incomplete information originated 
as tools to analyze games of this class. The first example was constructed by Zamir~(1971)~\cite{Zamir_sqrt_is_precise} %
(we include it in Subsection~\ref{subsect_notation}). Mertens and Zamir~(1976, 1995)~\cite{MZNormalGames, MZ1995}
proved that
for a class of games containing Zamir's example ${V_N}/{\sqrt{N}}$ converges to a limit related to the normal distribution.
They used a very technical analysis of a recurrent equation
for the sequence of values $V_N$. De~Meyer~(1996A, 1996B)~\cite{DeM1996_1, DeM1996_2}
introduced a duality approach which allowed to extend the results 
of Mertens and Zamir to a broader class of games without dealing with any technicalities. 
He obtained two representations of 
the ${V_N}/{\sqrt{N}}$ limit: one from the Central Limit Theorem and 
one as a solution of a partial differential equation.
Other results were based on finding explicit solutions (see Heuer~(1991)~\cite{Heuer}, % 
Domansky and Kreps~(1994, 1995, 1999)~\cite{DomKrepsEvent,DomKrepsMultinom,DomKrepsTransport}).
 Similar problems  of dynamic strategic use of information also demonstrate $\sqrt{N}$-behavior of the value. 
Mertens and Zamir~(1977)~\cite{MZvariation}
and De~Meyer~(1998)~\cite{DeM1998}
studied martingale optimization problem of the maximal $L^1$-variation
of a bounded martingale which is connected to the optimal speed of information revelation 
in games with $\sqrt{N}$-behavior (Gensbittel~(2015)~cite{CavU}
showed that any game can be reduced to a generalized problem of the maximal variation). 
Problems arising as continuous-time limits of $\sqrt{N}$-games were discussed by De~Meyer~(1999)~\cite{DeM1999}
and  Gensbittel~(2013)~\cite{Gensbittel2013}
(working with such limit problems is known as ``compact approach'', see Sorin~(2002)~\cite{Sorin}). %)

An example of another asymptotic behavior of $V_N$ is given by a financial market model introduced by  De~Meyer and Saley~(2003)~\cite{DeMSaley}
to analyze impact of information asymmetry on a financial market.
In this model a risky asset is exchanged between two agents; one of them is an insider and 
knows the actual value of a risky asset, and the other does not 
(see Subsection~\ref{subsect_notation} for the definition).
Depending on existence of the minimal currency unit the corresponding almost-fair game demonstrates 
either asymptotically bounded value or $V_N$  of the order of $\sqrt{N}$, though the model is robust to change 
of other modeling assumptions:
\begin{itemize}
\item If there is no minimal currency unit De~Meyer and Saley showed that the value has $\sqrt{N}$-behavior 
and the limit price-process is related to the Brownian motion. 
Prices are the elements of strategy, and the authors argue that their result provides 
endogenous justification for the Brownian motion in finance. 
The informal reason for the appearance of the Brownian motion is that 
the insider can benefit with revealing an arbitrary small  
``portion of information'' using an arbitrary small change of his price. As $N\to\infty$ optimal ``portions'' become infinitesimally small,
and the Brownian motion arises from the Central Limit Theorem,
and $\sqrt{N}$ originates from the standard normalization. Various extensions of this model 
(De~Meyer and Marino~(2004)~\cite{DeMMar2004}, %
De~Meyer~(2010), %~\cite{DeM2010}
De~Meyer and Fournier (2015)~\cite{DeM2015}) % 
 have the same properties.

\item Domansky~(2007)~\cite{Domansky2006}   
 and independently De~Meyer with Marino~(2005)~\cite{DeMMar}
showed that introducing the minimal currency unit radically changes 
the behavior of the model: the value $V_N$ becomes bounded as $N\to\infty$. 
The informal explanation is that the minimal currency unit does not allow insider to obtain non-zero profit revealing a small ``portion of information'' because he has to change his price significantly: at least by one currency unit. This forces him to reveal information fast and to lose the information advantage in a finite number of steps collecting only a bounded total gain. This informal explanation is supported by results of Domansky~\cite{Domansky2006} 
 who explicitly solved the infinite-stage version of the game and showed that the
optimal strategy of the insider reveals his private information in finite number of stages 
and that the price-process is a simple random walk over the lattice of admissible prices with absorption 
at the true value of the asset. 
Asymptotic boundedness and fast revelation holds for different modifications of the model with minimal currency unit (Domansky and Kreps (2009, 2013, 2016)~\cite{DomKreps2009, DomKreps2013, DomKreps2016}, %
Sandomirskaia (2016)~\cite{MarExponential}). %
\end{itemize}
The market model with minimal currency unit was historically the second non-trivial example of 
an  almost-fair game with bounded values. 
However, the first such example from (Domansky and Kreps~(1994)~\cite{DomKrepsEvent}) %
is similar to market model with three admissible bids
(see Subsection~\ref{subsect_notation} for details). 

Our research was inspired by the desire to find an abstract property responsible 
for the effect of $V_N$ being bounded in an almost-fair game.
Theorem~\ref{th_upper} shows that the value is bounded if in the non-revealing game 
Player~2 has an optimal strategy that is piecewise-constant as a function of the prior distribution. 
This is what we call the \emph{piecewise} property. The market models with minimal currency unit are piecewise games (Example~\ref{ex_marketispiecewise})
thereby Theorem~\ref{th_upper} immediately implies many results cited above.
Theorem~\ref{th_lower} shows the result is sharp: if for some interval of prior distributions the optimal strategy of Player~2 is unique and non-constant, then the value is of the order of $\sqrt{N}$. 
Remark~\ref{rm_comparison} compares this condition with the weakest known sufficient condition from De~Meyer~(1996A)~\cite{DeM1996_1}. %
Up to an additional assumption of uniqueness 
the theorems give a characterization of  possible asymptotic behaviors of $V_N$ for almost-fair games. These results are stated in Section~\ref{sect_results}.
The proofs explain the strategic origin of the two behaviors: we formalize 
the intuition that relates the behavior of $V_N$ with an opportunity to benefit from an arbitrary small revelation of information (see above and Propositions~\ref{prop_no_revelations} and~\ref{prop_small_revelations_imply_sqrt}). The proof of Theorem~\ref{th_upper} is contained in Section~\ref{sect_proof1}. The approach is based on ideas
of De~Meyer and Marino~\cite{DeMMar}
that link upper bounds on the value with invariant functions of the recurrent equation and on
the explicit construction of an invariant function using geometry of the Kantorovich metric.
Theorem~\ref{th_lower} is proved in Section~\ref{sect_proof2}. For this purpose we
pass from the game to a martingale optimization problem of generalized maximal variation (Gensbittel~(2013)~\cite{CavU}).  
By analyzing optimal strategies under small perturbations of matrix games 
we show that
the value of the martingale optimization problem is bounded from below by a quantity related to the typical 
movement of a simple random walk after $N$ steps
that is of the order of $\sqrt{N}$.  
The last Section~\ref{sect_open} presents a list of open problems.

\subsection{Notation, definitions, classic results, and examples}\label{subsect_notation}

In $N$\emph{-stage zero-sum game $\Gamma_N=\Gamma_N(p)$ with incomplete information on the side of Player~2} 
the players repeatedly play the same $\I\times \J$ matrix game $A^k$ depending on a random state $k\in\K$. 
The state $k$ is selected by chance before the first stage from the set of states $K$ according to a prior distribution $p\in\Delta(\K)$ known to players ($\Delta(\K)$ denotes the set of all probability measures over $\K$). Player~1 
knows the realization of $k$, but Player~2 
does not. 
The players choose their actions $i_n\in \I$ and $j_n\in \J$ at a stage $n=1,2,..N$ taking into account their current knowledge: the history of their actions $h_n=(i_t,j_t)_{t=1}^{n-1}$ (known to both) and knowledge of $k$ (for Player~1).  
Stage payoffs $A_{i_n,j_n}^k$ are not observed during the game. The objective of Player~1~(2) is to maximize~(minimize) the expected total  payoff\footnote{Usually one considers the expected average total payoff, i.e.,~the expected sum of stage gains divided by $N$, to ensure that  the sequence of values $V_N$ is bounded as $N\to\infty$. However, we do not follow this convention as $V_N$  remains bounded in the games we are interested in without any normalization.} $G_N=\expectation\left[\sum_{n=1}^N A_{i_n,j_n}^k\right]$. Players can randomize their actions using behavioral strategies (defined below), and so the expectation is taken with respect to the joint distribution of $k$ and $h_{N+1}$. 

\emph{A behavioral strategy} $\sigma$ of Player~1 is a collection $\{\sigma_n\}_{n=1}^N$, where $\sigma_n(h_n,k)\in\Delta(\I)$ is the distribution used by Player~1 to randomize his action $i_n$ depending on his ``knowledge'' $(h_n,k)$  at a stage $n$. In a behavioral strategy $\tau=\{\tau_n\}_{n=1}^N$ of Player~2 the distribution $\tau_n(h_n)\in\Delta(\J)$ used to select $j_n$ does not depend on the state $k$ as Player~2 does not know it. Note that $G_N=G_N(\sigma,\tau)$ represents the game $\Gamma_N$ in normal form. 
By Kuhn's theorem behavioral strategies are equivalent to mixed strategies if deterministic behavioral strategies (i.e., such that $\sigma_n$ and $\tau_n$ take values in the Dirac $\delta$-measures over $\I$ and $\J$, respectively) are considered as pure. Hence under the standard finiteness assumption of $\I,\J$, and $\K$ 
the min-max theorem applies. Therefore
the game  $\Gamma_N(p)$ has a value $V_N(p)$ and players have optimal strategies. 
The game also has a value for compact metric spaces $\I, \J$, and $\K$ with continuous $A$ by the standard approximation arguments: optimal strategies in the game
with discretized sets of actions and states induce $\varepsilon$-optimal strategies in the original one, and $\varepsilon$ goes to zero as discretization becomes finer (see Remark~\ref{rm_continuous_K} for an example of using such techniques).

The \emph{non-revealing game} $\Gamma^\nr_1(p)$ is a version of the one-stage game $\Gamma_1(p)$ where both players do not know the state $k$. 
This game is equivalent to $I\times\J$ matrix game with expected payoff matrix $A^p=\expectation A^k$. Hence $A^p=\sum_{l\in\K} p_l A^l$ for finite $\K$ (here $p_l=\mathbb{P}(\{k=l\})$ is the weight that $p$ gives to an element $l\in\K$). 
The value of the non-revealing game is traditionally denoted by $u(p)$. 

The $\cav[u]$-theorem of Aumann and Maschler shows that the the non-revealing game is responsible for the leading term of repeated game's value as $N\to\infty$: 
\begin{equation}
\label{eq_cavu}
V_N(p)=N\cdot \cav[u](p)+O(\sqrt{N}),\quad  N\to\infty,
\end{equation}
 in the case of finite $\I,\J$, and $\K$. Here $\cav[u]$ denotes the least concave majorant of $u$ treated as a function  $\Delta(\K)\to\R$.

We say that a repeated game $\Gamma_N$ is \emph{almost-fair} if the value $u(p)$ of the non-revealing game is zero for any $p\in\Delta(\K)$. In other words, almost-fair game is a game that becomes fair if Player~1 forgets $k$.
By~(\ref{eq_cavu}) the value of an almost-fair game can not grow faster\footnote{For almost-fair games with infinite $\I$, $\J$, and $\K$ the value can also grow as $N^\alpha$ with $\alpha\in(0.5,1)$, see Sandomirskiy~(2014)~\cite{Fedor_IJGT}. 
} than $\sqrt{N}$.

The following almost-fair games will be used to illustrate our results:
\begin{example}[\emph{The first game with $V_N$ of the order of $\sqrt{N}$}]\label{ex_1} 
Zamir~(1971)~\cite{Zamir_sqrt_is_precise}
showed that for the repeated game with $\I=\J=\K=\{0,1\}$ and stage games given by
\begin{equation}
\label{eq_Zamirs_example}
A^0=\begin{pmatrix}
3 & -1 \\ -3 & 1 
\end{pmatrix} \quad \mbox{and} \quad 
A^1=\begin{pmatrix}
2 & -2 \\ -2 & 2 
\end{pmatrix}
\end{equation}
the value grows as $\sqrt{N}$. Checking that the non-revealing game
$A^p=p_0A^0+p_1 A^1$ has zero value is left to the reader. 
\end{example}
\begin{example}[\emph{The first game with bounded $V_N$}]\label{ex_2} 
A repeated game is called \emph{flat} if $u(p)$ is an affine function of $p$.
Domansky and Kreps~(1994)~\cite{DomKrepsEvent}
explicitly solved all flat (``eventually revealing'' by their terminology) games with $\I=\J=\K=\{0,1\}$. They discovered that except for games similar to~\eqref{eq_Zamirs_example} this class contains almost-fair games with bounded values: if stage  games are
\begin{equation}
\label{eq_DK_example}
A^0=\begin{pmatrix}
1 & 0 \\ 0 & -(1-\alpha) 
\end{pmatrix} \quad \mbox{and} \quad 
A^1=\begin{pmatrix}
-1 & 0 \\ 0 & (1-\alpha) 
\end{pmatrix}, \quad \alpha\in[0,1],
\end{equation}
then the value $V_N(p)$ does not exceed $1/2$ for any $N$ and $p$. If $\alpha=0$, the
game $\Gamma_N(p)$ is especially simple: the optimal strategy of Player~1 does not depend on $n$ and prescribes to play ``top'' if $k=0$ and ``bottom'', otherwise. This strategy completely reveals $k$ after the first stage and leads to $V_N(p)=\max\{p_0,p_1\}$ for any $N$. For $\alpha>0$ the optimal strategy becomes $n$-dependent. Note that the degenerate case of $\alpha=0$ is equivalent to the discrete market model with $m=2$ described in Example~\ref{ex_3}.  

\end{example}

\begin{example}[\emph{The financial market model with asymmetric information}]\label{ex_3} 
A model of De~Meyer and Saley~(2003)~\cite{DeMSaley}
has $I=J=[0,1]$, the set of states $\K=\{0,1\}$, and the stage games given by $A_{i,j}^k=\sign(i-j) (k-\max\{i,j\})$. This repeated game can be interpreted as follows. The state $k$ is a liquidation value of a risky asset, and only Player~1, the insider, knows $k$. At each stage of $\Gamma_N(p)$ both players propose their prices $i_n$ and $j_n$ for the asset, and the player with higher price buys one unit of the asset from his opponent for this price. Players have enough assets and money. The objective of both players is to maximize their expected welfare after $N$ trading rounds. This game is almost-fair (by symmetry) and has $\sqrt{N}$-behavior of the value. We refer to this model as the continuous one.

Introducing the minimal currency unit $\frac{1}{m}, m\in \N,$ leads to the discrete version of the model with $I=J=\left\{0,\frac{1}{m},\frac{2}{m},..1\right\}$.
For the discrete model the value $V_N$ remains bounded as $N\to\infty$, see Domansky~(2007)~\cite{Domansky2006} 
and De~Meyer with Marino~(2005)~\cite{DeMMar}. % 
\end{example}

\section{Results}\label{sect_results}
The following property turns out to be responsible for 
the sequence of values $V_N$ being bounded as $N\to\infty$.
\begin{definition}
\label{def_trigger}
We say that a repeated game $\Gamma_N$ with incomplete information is a \emph{piecewise game} if there exists a function $y^*:\ \Delta(\K)\to\Delta(\J)$ taking a finite number of different values and such that $y^*(p)$ is an optimal strategy of Player~2 in the non-revealing game $\Gamma_1^\nr(p)$ for any $p\in\Delta(\K)$.
\end{definition}
\begin{theorem}\label{th_upper}
Let $\Gamma_N$ be an {almost-fair piecewise} repeated game with incomplete information  and finite $\I,\J$, and $\K$. For any $p\in\Delta(\K)$ and $N\geq 1$
\begin{equation}
\label{eq_theorem1}
0\leq V_N(p)\leq \|A\|_\lip Q,  
\end{equation}
where $\|A\|_\lip=\max_{i,j,k,k'}\left|A_{i,j}^k-A_{i,j}^{k'}\right|$ and $Q$ is the number of different values that $y^*$ takes (see Definition~(\ref{def_trigger})).
\end{theorem}
\begin{remark}
Note that only the upper bound in~(\ref{eq_theorem1}) is non-trivial since playing his optimal strategy for $\Gamma_1^\nr(p)$ at all stages of $\Gamma_N(p)$, Player~1 guarantees expected payoff of at least $0$.
\end{remark}
\begin{corollary}
Under the conditions of Theorem~\ref{th_upper} the sequence $V_N(p)$
has a finite limit  $V_\infty(p)$ as $N\to\infty$. Indeed, this 
sequence is bounded and non-decreasing  (if at the first stage of $\Gamma_{N+1}(p)$ Player~1  plays the strategy optimal in $\Gamma_1^\nr(p)$ and then plays as if the remaining game is $\Gamma_N(p)$, he gets at least $V_N(p)$).

\end{corollary}
\begin{remark}
Theorem~\ref{th_upper} can be easily extended in two ways:
\begin{itemize}
\item The result holds for flat games (games with affine $u(p)=\sum_{k\in K} p_k \val[A^k]$) if we put the error term $V_N(p)-N u(p)$ instead of $V_N(p)$ in~\eqref{eq_theorem1}. This can be deduced from strategic equivalence of any flat game to an almost-fair game with stage payoffs given by $\widetilde A^k=A^k-\val[A^k]$.
\item The result holds for games with arbitrary compact metric space $\K$,
finite $\I,\J$, and continuous $A$ without any changes. Indeed, the right-hand side of~\eqref{eq_theorem1} does not explicitly depend on cardinality of $\K$ that allows to use standard approximation techniques. For the formal argument see~Remark~\ref{rm_continuous_K}.

\end{itemize}
\end{remark}

\begin{example}
Let us check that the game from Example~\ref{ex_2} is a piecewise game. In a non-revealing game $A^p=p_0A^0+p_1 A^1$ the optimal strategy of Player~2 is to play ``left'' if $p_1> p_0$ and play ``right'' if $p_1<p_0$. If $p_1=p_0$ any convex combination of the two actions is optimal. So 
$$y^*(p)=\left\{\begin{array}{cc} \delta_L, & p_1\geq p_0\\ \delta_R, & p_1<p_0\end{array}\right.$$
does the job (here and below $\delta_a$ denotes the Dirac $\delta$-measure 
concentrated at $a$). 
\end{example}
\begin{example} \label{ex_marketispiecewise}
It is easy to check that the discrete market model (see Example~\ref{ex_3}) has the piecewise property for all $m$. If nobody knows $k$, and the probability $p_1=\mathbb{P}(\{k=1\})$ belongs to the interval $\left[\frac{q}{m},\frac{q+1}{m}\right]$ for some $q=0,1,..m-1$, then for both players it is optimal to select $i=j=\frac{q}{m}$. Hence the pure strategy $y^*(p)=\delta_{\frac{[p_1m]}{m}}$ guarantees $0$ in the non-revealing game ($[x]$ denotes the integral part of $x\in\R$). Informally, the reason is that, if nobody knows the value $k$ of a risky asset, the optimal price in the continuous model would be the expectation $\expectation k=p_1$, but in the discrete one only discrete prices are allowed, and hence players select the closest point to $p_1$. So the convergence results from Domansky~(2007)~\cite{Domansky2006} 
and De~Meyer with Marino~(2005)~\cite{DeMMar} 
become immediate corollaries of Theorem~\ref{th_upper}.

For the continuous market model (recall that it has $V_N$ of the order $\sqrt{N}$) the optimal strategies in the non-revealing game are $i=j=p_1$. They are unique and depend continuously on the prior $p$. This suggests a form of converse to 
Theorem~\ref{th_upper}.
\end{example}

Let 
$[p',p'']$ denote the segment $\{\alpha p'+(1-\alpha)p''\mid \alpha\in[0,1]\}$ for $p',p''\in\Delta(\K)$. By $\Delta^\mathrm{relint}(\K)$ we denote the relative interior of $\Delta(\K)$, i.e.,
 the set of all $p\in\Delta(\K)$ such that $p_k>0$ for any $k\in\K$.
\begin{theorem}\label{th_lower}
Let $\Gamma_N$ be an {almost-fair} repeated game with incomplete information and finite $\I,\J$, and $\K$. Suppose that there exists a segment $[p',p'']$ such that for all $p\in[p',p'']$ the optimal strategy $y^*(p)$ of Player~2 in the non-revealing game $\Gamma_1^\nr(p)$ is unique, but $y^*(p)$ takes infinitely many different values when $p$ ranges over $[p',p'']$. Then for any $p\in\Delta^\mathrm{relint}(K)$ 
\begin{equation}
\label{eq_theorem2} 
C_1\sqrt{N}\leq V_N(p)\leq C_2\sqrt{N}, \quad N\geq1,
\end{equation}
where $C_h=C_h(p)$, $h=1,2$, are positive constants independent of $N$.
\end{theorem}

\begin{remark}
The upper bound in Theorem~\ref{th_lower} follows from  the estimate~(\ref{eq_cavu}) of Aumann and Maschler. The lower bound shows that the piecewise property is almost a criterion of boundedness. We expect that the piecewise property is a criterion, but the role of uniqueness assumption in Theorem~\ref{th_lower} remains a question for future research.
\end{remark}

\begin{example}
Zamir's game (Example~\ref{ex_1}) fulfills the assumptions of Theorem~\ref{th_lower}. The optimal strategy $y^*(p)$ of Player~2 in the non-revealing game $p_0A^0+p_1 A_1$ is unique and equals $\left(\frac{1+p_1}{4}, \frac{3-p_1}{4}\right)$. It continuously depends on $p$ and takes infinitely many values on any segment $[p',p'']$ with $p'\ne p''$. Hence Theorem~\ref{th_lower} implies $\sqrt{N}$-growth of $V_N(p)$.
\end{example}

\begin{remark}
Theorem~\ref{th_lower} can not be directly applied
to the continuous market model (Example~\ref{ex_3})
because finiteness of $\I$ and $\J$ is important for the proof. 
\end{remark}

\begin{remark}\label{rm_comparison}
Assumptions of Theorem~\ref{th_lower} are weaker than
the condition for the $\sqrt{N}$-behavior by De~Meyer~(1996A)~\cite{DeM1996_1}, % 
the weakest condition from the literature. Under the additional assumption
that certain partial-differential equation has regular solutions, he claims that the limit  of ${V_N}/{\sqrt{N}}$ as $N\to\infty$ exists for an almost-fair game if (1) all the sets $\I,\J$, and $\K$ are finite (2) both players have equal number of actions, i.e., $|\I|=|\J|$ (3) in the non-revealing game Player~1 has unique optimal strategy $x^*$ independent of $p\in\Delta(\K)$, and this strategy is completely mixed. Let us check that Theorem~\ref{th_lower}  applies to any such game except the degenerate case with $V_N(p)=0$ for any $N$ and $p$. 
Kaplansky lemma (see De~Meyer~(1996A)~\cite{DeM1996_1}, Lemma~4.1) %
implies that the optimal strategy $y^*(p)$ of the second player  is unique and completely mixed 
and rationally depends  on elements of $A^p$. Therefore $y^*(p)$ either fulfills the assumptions of Theorem~\ref{th_lower} or does not depend on $p$ at all.
In the last case $V_N(p)=0$ because Player~2 defends zero by playing $y^*$
at any stage of $\Gamma_N(p)$.

Let us modify Zamir's example (Example~\ref{ex_1}) by adding 
a fixed convex combination of rows to both matrices 
$A^k$ as a new pure strategy of Player~1. This new game is strategically equivalent to the initial one and has $V_N$ of the order of $\sqrt{N}$ that can be deduced also from Theorem~\ref{th_lower}, but De~Meyer's result becomes inapplicable because $|\I|\ne|\J|$. So Theorem~\ref{th_lower} gives a wider class of $\sqrt{N}$-games, but existence of the limit $\lim_{N\to\infty}{V_N}/{\sqrt{N}}$ is an open problem for this class.
\end{remark}

\section{Recurrent equation, the Kantorovich~metric, and proof of Theorem~\ref{th_upper}}\label{sect_proof1}
There are two different approaches to investigate repeated games with incomplete information: from the perspective of the 
informed player and from the perspective of his non-informed 
opponent. The first one leads to a martingale-optimization problems that arise in selecting the optimal rate of revealing information. 
The second one, the dual approach, is based on the concept of dual game introduced by  De~Meyer and on its recurrent structure. The problem of finding the optimal strategy of Player~2 has a flavor of  multi-criteria optimization: such a strategy does not depend on the state $k$ but has to be efficient for any $k\in \K$. 
This is made precise in an application of the approachability theory studying games with vector payoffs to the behavior of uninformed player in long games with incomplete information.

%\interfootnotelinepenalty=10000
\enlargethispage{0.3cm}

Our proof of Theorem~\ref{th_upper} is based on the first approach\footnote{I 
am grateful to Eilon~Solan for telling 
me about the paper of Mannor and Perchet~(2013)~\cite{MannorPer}
which allows to use the second approach to prove boundedness of $V_N$. 
They studied fast convergence 
of Blackwell's approachability procedure for repeated games with vector payoffs and found that, if a target set $B$ is a polytope approachable by a finite number of pure actions, then the approaching player has a strategy $\tau$ such that for any strategy $\sigma$ of the opponent and any $N$ the expected Euclidean distance between the average payoff $\vec g_N$ after $N$ rounds and $B$ is bounded from above by $\frac{C}{N}$ with some constant $C$. 
 
Let us apply this result to show boundedness of $V_N$ in almost-fair piecewise games using the standard link between behavior of Player~2 and the approachability  (see Mertens, Sorin, and Zamir (2015)~\cite{Bigbook}, % 
Section V.2.c.). It is enough to show that for any $p\in\Delta(\K)$ in $N$-stage game the uninformed player has a strategy $\tau$ such that for any strategy $\sigma$ of Player~1 the expected distance between the normalized
vector payoff $\vec g_N=\left(\frac{1}{N}\sum_{n=1}^N A^l_{i_n,j_n}\right)_{l\in\K}\in\R^\K$ and the target set $B=\{v\in\R^\K\mid v_l\leq 0 \ \forall l\in\K \}$ is bounded by 
 $\frac{C}{N}$. Let us call a game pure-piecewise if piecewise condition is fulfilled in pure actions (i.e.,
 $y^*$ takes values in $\J$, not in $\Delta(\J)$). For pure-piecewise games $B$ is approachable in pure actions, and the result of Mannor and Perchet applies and provides boundedness of $V_N$. 
 
It remains to check that any piecewise game can be reduced to a pure-piecewise. Indeed, define a new set of actions of Player~2 by $\widetilde\J=\J\cup y^*(\Delta(\K))$ (we add $Q$ actions) and the stage game by
$\widetilde A^k_{i,j}=A^k_{i,j}$ for $i\in\I,j\in\J$ and
$\widetilde A^k_{i,y}=\sum_{j\in\J}y_j A^k_{i,j}$ for $y\in y^*(\Delta(\K))$. The new repeated game is pure-piecewise because playing mixed action $y=y^*(p)$ becomes equivalent to selecting pure action $y\in\widetilde\J$. The values of both games coincide since addition of a convex combination of rows as a new row to a matrix game does not change its value.}.
The scheme resembles the one used by De~Meyer and Marino~(2005)~\cite{DeMMar} 
to derive an upper bound for the value of the discrete market model.
We start from the recurrent equation  $V_{N+1}(p)=T[V_N](p)$ in the form from Gensbittel~(2015)~\cite{CavU} 
and construct an explicit non-negative invariant function $h$ of the Shapley operator $T$ using the Kantorovich metric. Monotonicity ideas of De~Meyer and Marino imply that $h$ is an upper bound for $V_N$. This lets us almost avoid strategic analysis of the game but nonetheless get an explicit upper bound.

\subsection{Recurrent equation}\label{subsect_Bellman} Given a strategy $\sigma$ of Player~1 the process of information revealing by his actions is described by a sequence of posterior distributions $p^{(n)}\in\Delta(\K)$ of $k$ at a stage $n$, i.e., $p^{(n)}_{l}$ are defined as the conditional probability of $k=l$ given $h_n$. Posterior distribution of $k$ can be treated as a dynamic state variable of the game from the Player's~1 point of view since $p^{(n)}$ represents the beliefs of Player~2 about $k$ at a current stage. Denote the sequence of random variables $p^{(1)}, p^{(2)},..p^{(N+1)}$ by $p^{(n\geq 1)}$.

The sequence of random variables $\xi^{(1)},\xi^{(2)},..\xi^{(N)}$ is called a martingale of length $N$ adapted to the natural filtration (hereafter, a martingale) if the conditional expectation $\expectation\left[\xi^{(n+1)}\mid \xi^{(1)},\xi^{(2)},..\xi^{(n)}\right]$ equals $\xi^{(n)}$ for all $n=1,..N-1$.

The process of posterior distributions $p^{(n\geq 1)}$ is a martingale of length $N+1$ with values in $\Delta(\K)$ and with non-random $p^{(1)}=p$ (see Mertens, Sorin, and Zamir (2015)~\cite{Bigbook}, %
Section V.2.a.). 

Let $\M_p$ denote the set of all $\Delta(\K)$-valued martingales of infinite length with  $p^{(1)}=p$. Formally, elements of $\M_p$ are pairs consisting of a probability space and a martingale defined on this space.

From the early works of Zamir and Mertens the recurrent equation for the sequence of values is the central tool
to study asymptotic behavior of the value. 
Gensbittel~(2015)~\cite{CavU}
represented the recurrent equation as a martingale 
optimization problem where Player~1 decides how much
information to reveal at the first stage by selecting the distribution of the uninformed player's beliefs at the second stage.
\begin{theorem}[Gensbittel (2015), Proposition~3.5]
\label{th_Gensbittel_Bellman}
For a game $\Gamma_N$ with finite $\I,\J$, and $\K$ the following recurrent relation holds for any $N\geq 0$ (by convention, $V_0\equiv 0$)
\begin{equation}
\label{eq_Bellman}
V_{N+1}(p)=T[V_N](p)=\sup_{p^{(n\geq 1)}\in\M_p}\left[\mathcal{V}_1(\mathcal{P})+\expectation V_{N}(p^{(2)})\right],
\end{equation}
where $\mathcal{P}\in \Delta(\Delta(\K))$ is the distribution of $p^{(2)}$,  and $\mathcal{V}_{N'}(\mathcal{P})$ is the value of the auxiliary ${N'}$-stage game $\G_{N'}(\mathcal{P})$ with partial information on the side of Player~1 defined below.
\end{theorem} 
The game $\G_{N'}(\mathcal{P})$ is a version of $\Gamma_{N'}$ where Player~1 is not fully informed of $k$ but receives a noisy signal such that his believes about $k$ are $\mathcal{P}$-distributed. It can be viewed as a usual game with incomplete information with $\Delta(\K)$ as a set of states, $\mathcal{P}$ as a prior distribution, and $A^p=\sum_{k\in\K}p_k A^k$ as a stage payoff function. By the min-max theorem  $\mathcal{V}_1$ has the following representation
\begin{equation}
\label{eq_nu1}
\mathcal{V}_1(\mathcal{P})=\min_{y\in\Delta(\J)}\expectation\left[\max_{i\in\I}\sum_{j\in\J}y_j A^{p^{(2)}}_{i,j}\right],
\end{equation}
 where $p^{(2)}\sim\mathcal{P}$, i.e., is $\mathcal{P}$-distributed.
 
\subsection{Non-existence of small profitable revelations}
The introductory section contains informal reasoning explaining that
the $\sqrt{N}$-behavior of the value is related with the opportunity for Player~1 to reveal his information by means of small beneficial portions. The purpose of this subsection is to formalize the opposite situation.

We say that in an almost-fair game $\Gamma_N$ \emph{Player~1 has no small profitable revelations} if $\Delta(\K)$ can be represented as a finite union of closed convex subsets $\Delta_q\subset\Delta(\K)$, $q=1,..Q'$, such that for any $q$ and any $\mathcal{P}$ supported on $\Delta_q$ the value $\mathcal{V}_1(\mathcal{P})$ is zero.

Informally, this definition says that, if the prior distribution $p$ is inside $\Delta_q$, then Player~1 can not benefit at first stage of $N$-stage game without changing the belief $p^{(2)}$ of Player~2 significantly enough
by pushing $p^{(2)}$ outside of $\Delta_q$ (see~\eqref{eq_Bellman}), i.e.,
without revealing enough information by his action $i_1$.

The following proposition supports the intuition that non-existence of small profitable revelations leads to bounded values.

\begin{proposition}\label{prop_no_revelations}
If in an almost fair-game $\Gamma_N$ with finite $\I,\J$, and $\K$
Player~1 has no small profitable revelations, then for any $N\geq 1$ and any $p\in\Delta(\K)$
$$V_N(p)\leq \|A\|_{\mathrm{lip}}Q'.$$
\end{proposition}

Theorem~\ref{th_upper} is a corollary of Proposition~\ref{prop_no_revelations} and the next lemma. 
Subsections~\ref{subsect_monotone},~\ref{subsect_Kantorovich}, and~\ref{subsect_invar_func} develop the technique to prove Proposition~\ref{prop_no_revelations} and thereby 
to complete the proof of Theorem~\ref{th_upper}.

\begin{lemma}\label{lm_no_revelations}
If a game $\Gamma_N$ fulfills the assumptions of Theorem~\ref{th_upper}, then Player~1 has no small profitable revelations with $Q'=Q$.
\end{lemma}
\begin{proof}
Denote by $\{y^q\}_{q=1}^Q\subset\Delta(\J)$ the set of all values that $y^*$ takes. Define $\Delta_q$ as the subset of all $p\in\Delta(\K)$  such that $y^q$ is an optimal strategy of Player~2 in $\Gamma_1^\nr(p)$. Then
$\Delta(\K)=\bigcup_{q=1}^Q\Delta_q$. The subset $\Delta_q$ is a closed convex polytope since this subset is cut from $\Delta(\K)$ by the family of linear inequalities $\sum_k p_k\left(\sum_{j} y_j^q A_{i,j}^k  \right)\leq 0,\quad i\in\I$. It remains to check that for any $q$ and 
any $\mathcal{P}\in\Delta(\Delta_q)$ we have $\mathcal{V}_1(\mathcal{P})=0$. Indeed, playing $y^q$ in $\G_1(\mathcal{P})$ Player~2 defends $0$. This follows from~(\ref{eq_nu1}) because  $\max_{i\in\I}\sum_{j\in\J}y_j^q A_{i,j}^{p^{(2)}}$ is the value of the non-revealing game $\Gamma_1^\nr(p^{(2)}) $ for $p^{(2)}\in\Delta_q$  and equals zero by almost-fairness. 
\end{proof}

\subsection{Monotonicity properties and the role of invariant functions}\label{subsect_monotone}
The Shapley operator $T$ defined by formula~\eqref{eq_Bellman} has the following properties:
\begin{enumerate}
\item \emph{Representation of the value: } $V_N=T^N[0]$ (recall that $V_0\equiv 0$);
\item \emph{Monotonicity:} if $f\geq g$, then $T[f]\geq T[g]$.
\item \emph{Increasing property (for almost-fair games): } $T[f]\geq f$.
\end{enumerate}
Here $f\geq g$ means $f(p)\geq g(p)$ for all $p\in\Delta(\K)$. The first two items immediately follow from~(\ref{eq_Bellman}). 
To prove the third item we take in~(\ref{eq_Bellman}) a constant martingale $p^{(n)}\equiv p$. For such martingale $\mathcal{P}=\delta_p$, where $\delta_p$ is the Dirac $\delta$-measure at $p$. The game $\G_1(\delta_p)$ can be identified with the non-revealing game $\Gamma_1(p)$ (see
Gensbittel~(2015)~\cite{CavU}). %)
Hence by almost-fairness $\mathcal{V}_1(\delta_p)=0$ that implies the third item.

Using the following observation, De~Meyer and Marino proved that the discrete market model has bounded values.
\begin{lemma}[De~Meyer, Marino~(2005)~\cite{DeMMar}, Lemma~4.3] %.
\label{lm_DeMMar}
If $T$ has the properties (1) and (2) mentioned above, $h\geq 0$, and $T[h]=h$, then   $V_N\leq h$ for any $N\geq 1$. 
\end{lemma}
Indeed, $V_N=T^N[0]\leq T^N[h]=h$ as $0\leq h$.

In order to prove Proposition~\ref{prop_no_revelations} we need to construct such $h$ for any almost-fair game where Player~1 has no small profitable revelations. We use the Kantorovich metric for that.

\subsection{The Kantorovich metric}\label{subsect_Kantorovich}
Let $(X_0,\dist_0)$ be a compact metric space. The metric $\dist_0$ induces the Kantorovich\footnote{Sometimes this metric is called the Wasserstein distance (named after Leonid Vaserstein). Note that Kantorovich introduced this metric to study optimal transportation problems 27 years before Vaserstein used it in dynamical system context (see the discussion by Vershik~(2013)~\cite{VershikKantorovich}.}  metric $\dist_1$ on $X_1=\Delta(X)$ by 
$$\dist_1(p',p'')=\inf_{x'\sim p',\ x''\sim p''} \expectation \left[\dist_0(x',x'')\right],$$
where infimum is taken over all joint distributions of $x'$ and $x''$ with marginals $p'$ and $p''$, respectively. This makes $X_1$ a compact metric space, and the definition can be iterated to define
$X_2=\Delta(X_1)$ and, more generally, $X_n$ for each 
 $n$.

The dual way to define $\dist_1$ is by the Kantorovich-Rubinstein formula
$$\dist_1(p',p'')=\sup_{|f(x')-f(x'')|\leq d_0(x',x'')} \int_{X_0}f(x)\,(dp'(x)-dp''(x)),$$
where supremum is over all real-valued functions $f$ that are $1$-Lipshitz with respect to $\dist_0$ (see~Villani~(2008)~\cite{Villani}, 5.16)
Hence for a Lipshitz function $g$ on $X_0$ the integral $\int_{X_0}g(x)\,dp(x)$ is Lipschitz with respect to $\dist_1$ as a function of $p$ with the same constant. 

Let us come back to games. Denote by $\dist_0$ the discrete metric on $\K$ defined by
$$\dist_0(k',k'')=\left\{\begin{array}{cc}1, & k'\ne k''\\ 0, & k'= k''\end{array}\right. .$$
It induces the Kantorovich metric $\dist_1$ that in this case coincides with the total-variation distance $\dist_1(p',p'')=\max_{B\subset\K}|p'(B)-p''(B)|$. In turn, $\dist_1$ induces the Kantorovich metric $\dist_2$ on $\Delta(\Delta(\K))$.
\begin{lemma}
\label{lm_nu_is_lipschitz}
For any game $\Gamma_N$ with finite $\I$, $\J$, and $\K$ 
\begin{equation}
\label{eq_Lipschitz_property}
|\mathcal{V}_1(\mathcal{P}')-\mathcal{V}_1(\mathcal{P}'')|\leq \|A\|_\lip \dist_2(\mathcal{P}',\mathcal{P}''),
\end{equation}
where $\|A\|_\lip$ is from Theorem~\ref{th_upper}.
\end{lemma}
This result with a constant $2\max_{i,j,k}|A_{i,j}^k|$ instead of $\|A\|_\lip$ is proved by Gensbittel~(2015)~\cite{CavU}, %
Proposition~2.1. For Zamir's game (Example~\ref{ex_1}) we have $2\max_{i,j,k}|A_{i,j}^k|=6$ and $\|A\|_\lip=1$.
\begin{proof}
The one-stage payoff $A_{i,j}^k$ is Lipschitz in $k$ with respect to $\dist_0$
with a constant $\|A\|_\lip$. Hence $A_{i,j}^p=\int_\K A_{i,j}^k\, dp(k)$ is $\|A\|_\lip$-Lipschitz in $p$
with respect to $\dist_1$. Also for any $y\in\Delta(\J)$ we get $\max_{i \in\I}\sum_{j\in\J} y_j A_{i,j}^p$ is $\|A\|_\lip$-Lipschitz in $p$. Together with~(\ref{eq_nu1}) this implies~(\ref{eq_Lipschitz_property}).
\end{proof}
\begin{remark}
The  above reasoning extends to the case of uncountable compact metric space $(\K,\dist)$ and $\|A\|_\lip$-Lipschitz payoffs by putting $\dist$ instead of $\dist_0$ in the proof above. Estimate~\eqref{eq_Lipschitz_property} can be also generalized to the case of $\mathcal{V}_N$ with arbitrary $N$. Indeed, $\sum_{n=1}^N A^p_{i_n,j_n}$ is an $N\|A\|_\mathrm{lip}$-Lipschitz function of $p$ for any fixed sequence of actions, and taking expectation or minimum/maximum with respect to side variables does not change the Lipschitz constant.
\end{remark}

\subsection{Construction of the invariant function and the end of proof}\label{subsect_invar_func}
The construction of the invariant function $h$ of $T$ is based on the next lemma. 
\begin{lemma}
\label{lm_nu_trigger}
Let  $\Gamma_N$ be an almost-fair game with finite $\I,\J,\K$ and without small profitable revelations
(i.e., it satisfies the assumptions of Proposition~\ref{prop_no_revelations}).
Then for all $q=1,..Q'$ and $\mathcal{P}\in\Delta(\Delta(\K))$
\begin{equation}
\mathcal{V}_1(\mathcal{P})\leq \|A\|_\lip\expectation_{p\sim\mathcal{P}}\, \dist_1(p,\Delta_q),
\end{equation}
where the distance from a point to a set is defined in the usual way
as $\dist_1(p,\Delta_q)=\inf_{p_q \in \Delta_q} \dist_1(p,p_q)$.
\end{lemma}
\begin{proof}
Since for any $\mathcal{P}'$ supported on $\Delta_q$ we have $\mathcal{V}_1(\mathcal{P}')=0$, the Lipshitz property~(\ref{eq_Lipschitz_property}) implies
$$\mathcal{V}_1(\mathcal{P})\leq \|A\|_\lip \dist_2(\mathcal{P},\mathcal{P}').$$
Let $R:\ \Delta(\K)\to\Delta(\K)$ be a continuous selection of $\argmin_{p'\in \Delta_q}\left[\dist_1(p,p')\right]$.
Picking $\mathcal{P}'$ equal to the push-forward of $\mathcal{P}$ by $R$ leads to
$\dist_2(\mathcal{P},\mathcal{P}')=\expectation_{p\sim\mathcal{P}} \dist_1(p,\Delta_q)$ and completes the proof. The intuition behind such a choice of $\mathcal{P}'$ is that we want to transport the ``portion of $\mathcal{P}$'' at each $p$ to the closest $p'\in\Delta_q$, and one can show that this choice is the optimal one.
\end{proof}
\begin{lemma}
\label{lm_invariant_h}
Under the assumptions of Proposition~\ref{prop_no_revelations}
$$h(p)=\|A\|_\lip\sum_{q=1}^{Q'}(1-\dist_1(p,\Delta_q))$$
defines a non-negative invariant function of $T$.
\end{lemma}
\begin{proof}
The total-variation metric $\dist_1$ is bounded by $1$, and, therefore, $h(p)$ is non-negative. Consider $T[h](p)$ for some $p$ from $\Delta_w$, $w=1,..Q'$. Fix a martingale $p^{(n\geq 1)}\in\M_p$. From Lemma~\ref{lm_nu_trigger} and the definition of $h$ we get
$$\mathcal{V}_1(\mathcal{P})+\expectation h(p^{(2)})\leq \|A\|_\mathrm{lip}\expectation\left[ \dist_1(p^{(2)},\Delta_w)+\sum_{q=1}^{Q'}(1-\dist_1(p^{(2)},\Delta_q))\right]=$$
\begin{equation}
\label{eq_estimates_with_h}
=\|A\|_\mathrm{lip} \expectation\left[1+\sum_{q\ne w}(1-\dist_1(p^{(2)},\Delta_q))\right].
\end{equation}
Note that the total-variation distance to a convex set is convex. Since $\expectation p^{(2)}=p$, Jensen's inequality applied to~(\ref{eq_estimates_with_h}) implies
$$\mathcal{V}_1(\mathcal{P})+\expectation h(p^{(2)})\leq\|A\|_\lip \left(1+ \sum_{q\ne w}^{Q'}(1-\dist_1(p,\Delta_q))\right).$$
Since $\dist_1(p,\Delta_w)=0$ we can return $1=1-\dist_1(p,\Delta_w)$ into the sum and get
$$\mathcal{V}_1(\mathcal{P})+\expectation h(p^{(2)})\leq \|A\|_\lip \sum_{q=1}^{Q'}(1-\dist_1(p,\Delta_q))=h(p).$$
Taking maximum over $p^{(n\geq 1)}$ we obtain $T[h](p)\leq h(p)$, and so $T[h]\leq h$. But the increasing property of $T$ says that $T[h]\geq h$. Thus $T[h]=h$. 
\end{proof}
Now Proposition~\ref{prop_no_revelations} becomes a combination of Lemmas~\ref{lm_DeMMar} and~\ref{lm_invariant_h}. This also completes the proof of Theorem~\ref{th_upper}.

\begin{remark}
\label{rm_continuous_K}
Let us check that the statement of Theorem~\ref{th_upper} remains 
valid for games with a compact metric space $\K$, finite $\I,\J$, and continuous $A$. 
Along the lines of the proof of  inequality~\eqref{eq_theorem1}
we show existence of the value $V_N$ for such a game.

 Fix $\delta>0$. By uniform continuity of $A$ we find $\varepsilon$-net $\K_\varepsilon$ of $\K$ such 
that for any $k\in\K$ and any $i,j$ we have 
\begin{equation}
\label{eq_approx}
\left|A^k_{i,j}-A^{k_\varepsilon}_{i,j}\right|\leq \delta,
\end{equation}
where $k_\varepsilon\in\K_\varepsilon$ denotes a point nearest to $k$. Enumerate points in $\K_\varepsilon$ and make $k_\varepsilon$ single-valued by selecting the point with 
the smaller number in case of a tie. 
Let $p^\varepsilon\in\Delta(\K_\varepsilon)$ be the distribution of $k_\varepsilon(k)$ if $k$ is distributed according to $p$.
Inequality~\eqref{eq_approx} implies that for any strategy $\sigma$ of Player~1 the guaranteed payoffs in $\Gamma_N(p)$ and in $\Gamma_N(p^\varepsilon)$ differ at most by $N\delta$.  The same holds for any strategy $\tau$ of Player~2. 
The game $\Gamma_N(p^\varepsilon)$ can be considered as a game  with finite set of states $\K_\varepsilon$,  hence has a value $V_N(p^\varepsilon)$ that does not exceed $Q\|A\|_\mathrm{lip}$ by Theorem~\ref{th_upper}. For
lower and upper values of $\Gamma_N(p)$ we have
$$V_N(p^\varepsilon)-N\delta\leq \underline{V}_N(p)\leq \overline{V}_N(p)\leq {V}_N(p^\varepsilon)+N\delta\leq  Q\|A\|_\mathrm{lip}+N\delta.$$
Since $\delta$ is arbitrary, the value exists and the inequality~\eqref{eq_theorem1} holds.
\end{remark}

\section{Small profitable revelations, simple random walks, parametric families of matrix games, and proof of Theorem~\ref{th_lower}}\label{sect_proof2}

The main tool in the proof is a representation for $V_N$ as the value of a martingale-optimization problem from~Proposition~3.5 of Gensbittel~(2015)~\cite{CavU}  
(definitions of $\M_p$ and $\mathcal{V}_1$ are introduced in Subsection~\ref{subsect_Bellman} where we already formulated another version of his result):
\begin{equation}
\label{eq_mart_opt}
V_N(p)=\max_{p^{(n \geq 1)}\in\M_p}\expectation\left[\sum_{n=1}^N \mathcal{V}_1(\mathcal{P}_n)\right],
\end{equation}
where $\mathcal{P}_n$ is the conditional distribution of $p^{(n+1)}$ given $p^{(1)},p^{(2)},..p^{(n)}$. 
\begin{remark} Starting from an arbitrary martingale
$p^{(n \geq 1)}$ one can explicitly construct a strategy of Player~1
that guarantees $\expectation\left[\sum_{n=1}^N \mathcal{V}_1(\mathcal{P}_n)\right]$ in $\Gamma_N(p)$, see Gensbittel~(2015)~\cite{CavU}. %~\cite{CavU}. 
This idea together with~(\ref{eq_mart_opt}) comes from De~Meyer~(2010)~\cite{DeM2010}
where it was formulated in the context of market models.

In particular, a strategy that gives to Player~1 an expected gain of the order of $\sqrt{N}$ can be extracted from our proof of Theorem~\ref{th_lower}.
\end{remark}

We begin with formalization of what does it mean that Player~1 can benefit by revealing an arbitrary small amount of information. Intuition mentioned in the introduction for the continuous market model connects this property with $\sqrt{N}$-behavior. So does Proposition~\ref{prop_small_revelations_imply_sqrt} below. It is proved in Subsection~\ref{subsec_random_walk} using formula~\ref{eq_mart_opt} and explicit construction 
of a martingale that provides $\sqrt{N}$-lower bound on
$V_N(p)$. In Subsection~\ref{subsect_parametric} we check that the assumptions of Theorem~\ref{th_lower} imply existence of small profitable revelations.

\subsection{Small profitable revelations imply $\sqrt{N}$-behavior}
In an almost fair game Player~1 is said  to have \emph{ small profitable revelations} if for some $p',p''\in\Delta(\K)$, $p'\ne p''$,
there exists $C_A>0$ such that for any $\alpha_1,\alpha_2\in[0,1]$ 
$$\mathcal{V}_1\left(\frac{1}{2}\delta_{p(\alpha_1)} + \frac{1}{2}\delta_{p(\alpha_2)}\right)\geq C_A |\alpha_1-\alpha_2|,$$
where $p(\alpha)=\alpha p'+ (1-\alpha)p''$, and 
$\delta_p\in\Delta(\Delta(\K))$ is the Dirac $\delta$-measure at $p$.

In other words, Player~1 has small profitable revelations if 
at first stage of $N$-stage game
by an arbitrary small change of the uninformed player's belief (from $p^{(1)}=p\left(\frac{\alpha_1+\alpha_2}{2}\right)$ to $p^{(2)}=p(\alpha_1)$ or to $p^{(2)}=p(\alpha_2)$ equally likely) he can get a profit of the order of this change.

\begin{proposition}\label{prop_small_revelations_imply_sqrt}
If in an almost-fair game $\Gamma_N$ with finite $\I,\J$, and $\K$ Player~1 has small profitable revelations, then $V_N(p)$ is of the order of $\sqrt{N}$ for any $p\in\Delta^\mathrm{relint}(\K)$. 
\end{proposition}
This proposition is proved in the next subsection. 
Together with the next proposition it implies Theorem~\ref{th_lower}.
\begin{proposition}
\label{prop_splitprofit}
Under the assumptions of Theorem~\ref{th_lower} Player~1 has small profitable 
revelations.
\end{proposition} 
We prove Proposition~\ref{prop_splitprofit} in Subsection~\ref{subsect_parametric}.

\subsection{Simple random walks and proof of Proposition~\ref{prop_small_revelations_imply_sqrt}}\label{subsec_random_walk}
In order to prove  Proposition~\ref{prop_small_revelations_imply_sqrt}
we define a martingale $p^{(n\geq 1)}$ from $\M_{p\left(\frac{1}{2}\right)}$ that ensures $\sqrt{N}$-lower-bound
in~\eqref{eq_mart_opt} by 
$$p^{(n)}=p\left(\frac{1}{2}+\frac{Z^{(n\wedge\tau_N)}}{4\sqrt N}\right),$$ where $Z^{(n\geq 1)}$ is the simple random walk over $\Z$ starting from $0$ at time $n=1$, stopping time $\tau_N$ is the minimal $n$ such that $|Z^{(n)}|\geq 2\sqrt{N}-1$, and $\wedge$ denotes taking minimum. As above, $p(\alpha)=\alpha p'+(1-\alpha)p''$. 

The following lemma says that $p^{(n\geq 1)}$ makes enough jumps before time $N+1$.
\begin{lemma}
\label{lm_srw}
$\P(\{\tau_N> N\})>\frac{1}{2}$ for any $N\geq 1$.
\end{lemma}
\begin{proof}
Let us estimate $\P(\{\tau_N> N\})= 1-\P\left(\{\max_{n=1,..N}\left|Z^{(n)}\right|\geq 2\sqrt{N}-1\}\right)$ from below. 

For any square-integrable martingale $\xi^{(n\geq 1)}$, $\lambda\ne 0$, and $N\geq 1$ the maximal Doob inequality holds (see Revyz, Yor~(1999)~\cite{Yor}, 
Corollary~(1.6) from Chapter~II): $$\mathbb{P}\left(\left\{\max_{n=1,..N} \left|\xi^{(n)}\right|\geq \lambda\right\}\right)\leq \frac{\expectation\left[\left(\xi^{(N)}\right)^2\right]}{\lambda^2}.$$
The simple random walk $Z_N$ is a martingale, and hence by Doob's inequality 
$$\P\left(\left\{\max_{n=1,..N}\left|Z^{(n)}\right|\geq 2\sqrt{N}-1\right\}\right)\leq \frac{N-1}{(2\sqrt{N}-1)^2}<\frac{1}{2}.$$ Thus $\P(\{\tau_N> N\})>\frac{1}{2}$ that completes the proof of the lemma.
\end{proof}

\begin{proof}[Proof of Proposition~\ref{prop_small_revelations_imply_sqrt}]
First, let us prove the result for $p=p\left(\frac{1}{2}\right)$.
Using the constructed martingale $p^{(n\geq 1)}$ in formula~(\ref{eq_mart_opt}) we get
$$V_N\left(p\left(\frac{1}{2}\right)\right)\geq \expectation\left[\sum_{n=1}^N \mathcal{V}_1(\mathcal{P}_n)\right]\geq \sum_{n=1}^N \expectation\left[\mathcal{V}_1(\mathcal{P}_n)\mid \tau_N>n \right]\P\left(\{\tau_N>n\}\right).$$
If the martingale has a jump at time $n$, i.e., if $\tau_N>n$, then
the conditional distribution $\mathcal{P}_n$ of $p^{(n+1)}$ given $p^{(1)},..p^{(n)}$ with $p^{(n)}=p(\alpha)$ is equal to
$\frac{1}{2}\delta_{p\left(\alpha+\frac{1}{4\sqrt N}\right)} + \frac{1}{2}\delta_{p\left(\alpha-\frac{1}{4\sqrt N}\right)}$. Since Player~1 has small profitable revelations
 $\expectation\left[\mathcal{V}_1(\mathcal{P}_n)\mid \tau_N>n \right]\geq
 \frac{C_A}{2\sqrt{N}}$. Therefore
 $$V_N\left(p\left(\frac{1}{2}\right)\right)\geq 
\frac{C_A}{2\sqrt{N}} \sum_{n=1}^N \P\left(\{\tau_N>n\}\right)\geq \frac{C_A}{2}\sqrt{N}\P\left(\{\tau_N>N\}\right)\geq \frac{C_A}{4}\sqrt{N}.$$
(in the last inequality we applied Lemma~\ref{lm_srw}).  This gives the result for $p=p\left(\frac{1}{2}\right)$. 
 
Now consider an arbitrary prior distribution $p$ from the relative interior $\Delta^{\mathrm{relint}}(\K)$ of the simplex. We need two observations:
\begin{itemize}
\item For almost-fair games $V_N(\delta_k)=0$ for any $k$. Indeed, for such prior both players know that the game is $N$ times repeated matrix game $A^k$ that can be identified with $\Gamma_1^\nr(\delta_k)$ and so has zero value.
\item $V_N$ is a concave function of $p$ (see Mertens, Sorin, Zamir~(2015)~\cite{Bigbook}, Section V.1.,  % 
or Gensbittel~(2015)~\cite{CavU}). %).
\end{itemize}
Any $p\in\Delta^{\mathrm{relint}}(\K)$  can be represented as a convex combination $p=\beta p\left(\frac{1}{2}\right)+\sum_{k\in\K}\beta_k\delta_k$ with $\beta>0$.
The maximal possible $\beta$ equals $1-\dist_1\left(p,p\left(\frac{1}{2}\right)\right)$, where $\dist_1$ denotes the total-variation distance ($\dist_1$ is defined in Subsection~\ref{subsect_Kantorovich}). Concavity leads to
$$V_N(p)\geq \beta V_N\left(p\left(\frac{1}{2}\right)\right)+\sum_{k\in\K}\beta_k V_N(\delta_k)=\beta V_N\left(p\left(\frac{1}{2}\right)\right)\geq \left(1-\dist_1\left(p,p\left(\frac{1}{2}\right)\right)\right)\frac{C_A }{4}\sqrt{N}.$$
Thus Proposition~\ref{prop_small_revelations_imply_sqrt} is proved. 
\end{proof}

\subsection{Parametric families of matrix games and proof of Proposition~\ref{prop_splitprofit}}\label{subsect_parametric}
The proof of  Proposition~\ref{prop_splitprofit} is based on two lemmas. 
The first one says that we can find a subinterval where the unique optimal strategy $y^*=y^*(\alpha)\in\Delta(\J)$ of Player~2 in the non-revealing game $\Gamma_1^\nr(p(\alpha))$ depends on $\alpha$ strongly enough. 
Let us formulate this result more rigorously.
\begin{lemma}
\label{lm_optimal_y_is_unstable}
Under the assumptions of Theorem~\ref{th_lower} there exists a constant 
$C>0$ and $0\leq\alpha_\mathrm{min}<\alpha_\mathrm{max}\leq 1$ such that
for any $\alpha_1,\alpha_2\in [\alpha_\mathrm{min}, \alpha_\mathrm{max}]$ 
\begin{equation}
\label{eq_optimal_y_is_unstable}
\|y^*(\alpha_1)-y^*(\alpha_2)\|\geq C|\alpha_1-\alpha_2|,
\end{equation}
where $\|\,\cdot\,\|$ denotes the Euclidean norm.
\end{lemma}
\begin{proof}
Consider an arbitrary matrix game $A$. Let $y^*$ be an extreme point
of the set of Player's~2 optimal strategies. Denote by $y^{*>0}$ the vector of its non-zero components. 
The classical result of Snow and Shapley (see Karlin~(1959)~\cite{Karlin}, %
Theorem~2.4.3) says that there exists a square submatrix $M$ of $A$ such that
$y^{*>0}=\frac{\mathrm{adj}(M)e}{\langle e,\mathrm{adj}(M) e\rangle}$ with   non-zero denominator. Here $e$ denotes the vector of all ones, $\langle\cdot,\cdot\rangle$ is the standard scalar product, and $\mathrm{adj}(M)$ is the adjugate of $M$, i.e., $\mathrm{adj}(M)_{i,j}$
equals the $(j,i)$-cofactor of $M$. In particular $y^*$ is a rational function of entries of $M$.

The result of Snow and Shapley implies that, if a matrix game $A=A(\alpha)$ is an affine function of a parameter $\alpha$ and the optimal strategy $y^*=y^*(\alpha)$ of Player~2 is unique for all $\alpha$, then $y^*$ is a piecewise-rational function of $\alpha$ with finitely many domains of rationality (since the number of square submatrices of $A$ is finite).

In Theorem~\ref{th_lower} it is assumed that $y^*(\alpha)$ takes infinitely many different values. Therefore there is an interval $D\subset[0,1]$ such that $y^*(\alpha)$ is a non-constant rational function on $D$. Hence there is a point in the interior of $D$ where $\frac{d}{d\alpha}y^*\ne 0$. In a small annulus $[\alpha_\mathrm{min},\alpha_\mathrm{max}]$ of this point $y^*$ is close to affine function with non-zero slope. This implies~(\ref{eq_optimal_y_is_unstable}).
\end{proof}
The next lemma says that losses of Player~2 from playing non-optimal mixed strategy  in the non-revealing game are proportional to the distance between the strategy used and the optimal one. This result is similar to Lemma~4.3 from De~Meyer~(1996A)~\cite{DeM1996_1}. %. 
We borrow the idea of the proof from his paper.
\begin{lemma}
\label{lm_errors_role}
Under the conditions of Lemma~\ref{lm_optimal_y_is_unstable} there exist an interval $[\alpha_\mathrm{min}', \alpha_\mathrm{max}']\subset [\alpha_\mathrm{min}, \alpha_\mathrm{max}]$
of non-zero length and a constant 
$C'>0$ such that for any $y\in\Delta(\J)$ and any $\alpha\in[\alpha_\mathrm{min}', \alpha_\mathrm{max}']$
\begin{equation}
\label{eq_errors_cost}
\max_{i\in\I} \sum_{j\in\J}A^{p(\alpha)}_{i,j} y_j\geq C' \|y^*(\alpha)-y\|.
\end{equation}
\end{lemma}
\begin{proof}
Let $\I_0(\alpha)$ be the set of all $i\in\I$ such that $\sum_{j\in\J}A^{p(\alpha)}_{i,j} y^*_j(\alpha)=0$. Find a subinterval
$[\alpha_\mathrm{min}', \alpha_\mathrm{max}']$  of $[\alpha_\mathrm{min}, \alpha_\mathrm{max}]$ such that $\I_0(\alpha)$
and the support $\J^*(\alpha)$ of $y^*(\alpha)$ remain the same when $\alpha$ ranges over the subinterval. Such subinterval exists 
since $y^*(\alpha)$ is a continuous function of $\alpha$ (indeed, if $\alpha_m\to\alpha$ as $m\to\infty$ and $y=\lim_{m\to\infty} y^*(\alpha_m)$, then $y$ is an optimal
strategy in $A^{p(\alpha)}$, and hence $y=y^*(\alpha)$ by uniqueness). 
The left-hand side of~\eqref{eq_errors_cost} is bounded from below by
$f(\alpha, \varepsilon)=\max_{i\in\I_0(\alpha)}\sum_{j\in\J}A^{p(\alpha)}_{i,j} \varepsilon_j$, where $\varepsilon=y-y^*(\alpha)$. Define $D(\alpha)=\left\{\varepsilon\in\R^\J\mid \varepsilon=(y-y^*(\alpha))/\|y-y^*(\alpha)\|, \ y\in\Delta(\K)\setminus\{y^*(\alpha)\}\right\}$ and note that $D(\alpha)$ is a compact set and does not depend on $\alpha$ for $\alpha\in [\alpha_\mathrm{min}', \alpha_\mathrm{max}']$.
A function $f$ is a homogeneous function of $\varepsilon$ for any $\alpha$. Therefore it is enough to show that $C=\min_{[\alpha_\mathrm{min}', \alpha_\mathrm{max}']\times D} f(\alpha, \varepsilon)>0$. Since $f$ is continuous, this minimum exists an is attained at some $(\alpha^*,\varepsilon^*)$. If $C\leq 0$, then for
$\lambda>0$ small enough we would have $\max_{i\in\I} \sum_{j\in\J}A^{p(\alpha^*)}_{i,j} y_j\leq 0$ for $y=y^*(\alpha^*)+\lambda\varepsilon^*$
that contradicts the uniqueness of the optimal strategy of Player~2 in 
$\Gamma_1^\nr(p(\alpha^*))$. This contradiction completes the proof.
\end{proof}

Now Proposition~\ref{prop_splitprofit} can be proved easily.
\begin{proof}[Proof of Proposition~\ref{prop_splitprofit}]
Consider the game $\G_1(\mathcal{P})$ with $\mathcal{P}=\frac{1}{2}\delta_{p(\alpha_1)} + \frac{1}{2}\delta_{p(\alpha_2)}$. This game can be interpreted as follows. Chance selects one of two matrix games $A^{p(\alpha_1)}$ or $A^{p(\alpha_2)}$ equally likely. Then the game selected is played, but only Player~1 knows the choice. 
We assume that $\alpha_1,\alpha_2\in[\alpha_\mathrm{min}',\alpha_\mathrm{max}']$, where $\alpha_\mathrm{min}'$ and $\alpha_\mathrm{max}'$ come from Lemma~\ref{lm_errors_role}.
Suppose Player~2 uses a strategy $y\in\Delta(\J)$. Then Lemmas~\ref{lm_optimal_y_is_unstable} and~\ref{lm_errors_role} prevents him from being successful in both $A^{p(\alpha_1)}$ and $A^{p(\alpha_2)}$.
More formally, by playing $y$ he defends (see~(\ref{eq_nu1}))
$$\expectation_{p^{(2)}\sim\mathcal{P}}\left[\max_{i\in\I}\sum_{j\in\J}y_j A^{p^{(2)}}_{i,j}\right]=\frac{1}{2}\max_{i\in\I}\sum_{j\in\J}y_j A^{p(\alpha_1)}_{i,j}+\frac{1}{2}\max_{i\in\I}\sum_{j\in\J}y_j A^{p(\alpha_2)}_{i,j}.$$
Lemma~\ref{lm_errors_role} implies that this amount is bounded from below by 
$$\frac{1}{2}C' \|y^*(\alpha_1)-y\|+\frac{1}{2}C' \|y^*(\alpha_2)-y\|
\geq\frac{1}{2}C'\|y^*(\alpha_1)-y^*(\alpha_2)\|.$$
Here we applied the triangle inequality. By Lemma~\ref{lm_optimal_y_is_unstable} this quantity is greater than  
$C C'|\alpha_1-\alpha_2|/2.$ Since $y$ was arbitrary we get
$$\mathcal{V}_1\left(\frac{1}{2}\delta_{p(\alpha_1)} + \frac{1}{2}\delta_{p(\alpha_2)}\right)\geq\frac{1}{2}C C'|\alpha_1-\alpha_2|.$$
This inequality is equivalent to existence of small profitable revelations
with redefined $p'=p(\alpha_\mathrm{min}')$,  $p''=p(\alpha_\mathrm{max}')$,
and $C_A=\frac{C C'}{2(\alpha_\mathrm{max}'-\alpha_\mathrm{min}')}$
and concludes the proof of Proposition~\ref{prop_splitprofit}. 
\end{proof}

\section{What is next?}\label{sect_open}
The results of this paper raise more questions than
answers:
\begin{itemize}
\item The piecewise property can be easily checked for any particular repeated game because it appeals only to a parametric family of matrix games. But can one describe piecewise games more explicitly? How does one construct examples of such games?  The only way we know is to start from discrete market game and to modify it somehow.
\item Here we obtained that the sequence of values $V_N$ of $N$-stage almost-fair piecewise game converges to some finite limit $V_\infty$ as $N\to\infty$. This suggests one to consider an infinite-stage version $\Gamma_\infty$ of such game. Do players have optimal strategies in $\Gamma_\infty$? For discrete market games the positive answer with explicit construction is given by Domansky~(2007)~\cite{Domansky2006}. %. 
In particular, he obtained that in the infinite market game Player~1 (behaving optimally) will reveal his private information in a finite time. Do we have such a counterintuitive effect for all almost-fair piecewise games? What are the properties of $V_\infty$? For market games it is a piecewise-linear function of $p$, see the paper of Domansky and the one of De~Meyer and Marino~(2005)~\cite{DeMMar}. %. 
What is the speed of $V_N$ convergence? Sandomirskaia~(2016)~\cite{MarExponential} 
showed that the speed is exponential for the discrete market model.

\item The role of uniqueness assumption in Theorem~\ref{th_lower} should be clarified. For example, are there any almost-fair games with $V_N$ growing to infinity but slower than $\sqrt{N}$? Does there exist a non-piecewise almost-fair game with bounded values?

\item Existence of the limit ${V_N}/{\sqrt{N}}$ as $N\to\infty$ for games with $V_N$ of the order of $\sqrt{N}$ is 
an important open question.
For the widest class of games considered in the literature, the result is conditional: existence is proved
under the assumption that  a certain partial-differential equation
has sufficiently regular solution (De~Meyer~(1996A)~\cite{DeM1996_1}).
%).

\item To find a proper generalization of Theorem~\ref{th_upper} to infinite $\I$ and $\J$ and of Theorem~\ref{th_lower} to infinite
$\I, \J, \K$  seems to be interesting. So are  
generalizations of the theorems to games  that are not almost-fair or have incomplete information on both sides.

\end{itemize} 
We plan to discuss some of the questions raised
in subsequent publications.

\end{document}